\documentclass[11pt]{article}
\usepackage{graphicx}
\usepackage{amsmath, amssymb, url}

\usepackage[noline,noend,linesnumberedhidden,boxed]{algorithm2e}
\DontPrintSemicolon
\SetArgSty{textnormal}
\SetAlCapSkip{1ex}
\SetAlgoCaptionLayout{centerline}

\newtheorem{theorem}{Theorem}[section]
\newtheorem{lemma}[theorem]{Lemma}
\newtheorem{corollary}[theorem]{Corollary}
\newcommand{\qed}{\rule{7pt}{7pt}}
\newenvironment{proof}{\noindent {\it Proof\/}:}{$\qed$ \medskip}

\newcommand{\Ptime}{\mbox{P}}
\newcommand{\NP}{\mbox{NP}}

\DeclareMathOperator*{\argmax}{\arg\!\max}

\begin{document}

\title{Easy Capacitated Facility Location Problems,\\ with Connections to Lot-Sizing}
\author{Alice Paul\thanks{Address: Olin College, 1000 Olin Way, Needham, MA 02492, \tt{apaul@olin.edu}} \and
		David P.\ Williamson\thanks{Address: School of Operations Research and Information Engineering, Cornell
	University, Ithaca, NY 14853, \tt{davidpwilliamson@cornell.edu}}}


\date{}
\maketitle

\begin{abstract}
	In this note, we consider the capacitated facility location problem when the transportation costs of the instance satisfy the Monge property.  We show that a straightforward dynamic program finds the optimal solution when the demands are polynomially bounded.   When demands are not polynomially bounded, we give a fully polynomial-time approximation scheme by adapting an algorithm and analysis of Van Hoesel and Wagelmans \cite{VanHoeselW01}.
\end{abstract}

\section{Introduction}

In the {\em capacitated facility location problem}, we are given a set $D$ of clients, a set $F$ of facilities, a cost $f_i \geq 0$ for each $i \in F$ for opening facility $i$, a capacity $U_i$ for each facility $i \in F$, a demand $d_j > 0$ for each client $j \in D$, and a cost $c_{ij} \geq 0$ for transporting a single unit of demand from client $j \in D$ to a facility $i \in F$.  We allow $c_{ij} = \infty$, meaning that demand cannot be shipped from client $j$ to facility $i$.  We assume that all finite numbers in the input are integers.  The goal is to find a set $F'$ of facilities to open plus an assignment $x^*_{ij} \in [0,1]$ of demand to open facilities such that $x^*_{ij} > 0$ only when $i \in F'$, $\sum_{j \in D} x^*_{ij}d_j \leq U_i$ for all $i \in F'$, $\sum_{i \in F'} x^*_{ij} = 1$ for all $j \in D$, and we minimize the cost of opening facilities and transporting demand, $$\sum_{i \in F'} f_i + \sum_{i \in F', j \in D} c_{ij} d_j x^*_{ij}.$$  In what follows, we will let $m = |F|$ denote the number of facilities and $n = |D|$ denote the number of clients.  We say that the capacities are {\em uniform} if $U_i = U$ for all $i \in F$.  We say that the demands are {\em polynomially bounded} if $\sum_{j \in D} d_j$ can be bounded by some polynomial in the input size of the problem.

We consider exact polynomial-time algorithms and approximation algorithms for the problem.  An $\alpha$-approximation algorithm is an algorithm that runs in polynomial time and returns a solution of cost at most $\alpha$ times the cost of an optimal solution.  The value $\alpha$ is called the approximation ratio or the performance guarantee of the algorithm.  A polynomial-time approximation scheme (PTAS) is a family of algorithms $A_\epsilon$ parameterized by $\epsilon > 0$, such that each $A_\epsilon$ is a $(1+\epsilon)$-approximation algorithm.  A fully polynomial-time approximation scheme (FPTAS) is a PTAS in which the running time of each $A_{\epsilon}$ is polynomial in $1/\epsilon$.

In its full generality, the capacitated facility location problem is strongly NP-hard, and includes the set cover problem as a special case \cite{Hochbaum82}, and thus no better than an $O(\ln n)$-approximation algorithm is possible unless $\Ptime = \NP$ \cite{LundY94,BellareGLR93,Feige98, MoshkovitzR10, Moshkovitz15, DinurS14}.  Researchers have considered the special case in which transportation costs are {\em metric}, in the following sense: for any clients $j,k \in D$ and facilities $h, i \in F$, $$c_{ij} \leq c_{hj} + c_{hk} + c_{ik}.$$  In this case, a simple local search algorithm initially given by Korupolu, Plaxton, and Rajaraman \cite{KorupoluPR00} has been shown by Aggarwal, Louis, Bansal, Garg, Gupta, Gupta, and Jain \cite{AggarwalLBGGGJ13} to be a $(3+\epsilon)$-approximation algorithm in the uniform case; a less simple local search algorithm is a $(5 + \epsilon)$-approximation algorithm for the non-uniform case \cite{BansalGG12}.  An, Singh, and Svensson \cite{AnSS17} have given a linear programming relaxation of the problem such that an LP rounding algorithm is a 316-approximation algorithm for the problem.  A long line of research has given approximation algorithms with small performance guarantees in the case of the metric {\em uncapacitated} facility location problem, in which $U_i = \infty$ for all facilities $i$; the best currently known algorithm is a 1.488-approximation algorithm due to Li \cite{Li11}.  It is known that the metric uncapacitated facility location problem cannot have an $\alpha$-approximation algorithm for $\alpha < 1.463$ unless $\Ptime = \NP$ \cite{GuhaK99}.

Exact algorithms and approximation algorithms have been given in some other special cases motivated by problems in inventory management.  In {\em single-item capacitated lot-sizing problem}, each client $j$ represents a point in time $t(j) \in [1,T]$, with a demand $d_j$ for a given item.  Each facility $i \in F$ represents a point in time $t(i) \in [1,T]$ at which an order can be placed for up to $U_i$ units of the given item at a cost of $f_i$; a client can only receive items from a facility at the same or an earlier point in time, so that $c_{ij} = \infty$ if $t(i) > t(j)$.  The transportation cost $c_{ij}$ is in this case a per-unit {\em holding cost} for holding a unit of the given item in inventory from time $t(i)$ to $t(j)$; holding costs for items ordered at time $t(i)$ are assumed to increase as $t(j)$ increases. A typical assumption is that holding costs are {\em linear}: that for each unit of time $[t,t+1)$ there is a holding cost $h_t$, and the cost of holding an item in inventory from time $t$ to $t' > t$ is $\sum_{k=t}^{t'-1} h_k$.  So in the linear case $c_{ij} = \sum_{t=t(i)}^{t(j)-1} h_t$.  In the {\em multi-item} variant of the problem, there are different types of items and a single point in time $t$ can have demands for each type of item, and a different holding cost for each type of item.  The capacity $U_i$ of a facility is the total number of items (summed across all item types) that can be ordered at time $t(i)$.  The different item types can be modelled by creating a client $j$ with $t(j) = t$ for each different type of item that has a demand at time $t$.  In the multi-item case, the holding costs can be linear and identical for each item, linear and different for each item type $j$ (so that there is a holding cost $h_{jt}$ for each item $j$ and time period $t$), or simply {\em increasing}: all we know for a given demand for a given item $j$ at time $t(j)$ is that the cost $c_{ij}$ is nondecreasing as $t(i)$ decreases.

Florian, Lenstra, and Rinnooy Kan \cite{FlorianLRK80} show that the single-item capacitated lot-sizing problem is weakly NP-hard even if all holding costs are zero by reduction from the subset sum problem (see Bitran and Yannasse \cite{BitranY82} for further NP-hardness results).
Van Hoesel and Wagelmans \cite{VanHoeselW01} give an FPTAS for the single-item capacitated lot-sizing problem with increasing holding costs.  Anily and Tzur \cite{AnilyT05} give a dynamic program for the multiple item case in which the holding costs are linear and identical for all items, and ordering/facility costs are identical.  Anily, Tzur, and Wolsey \cite{AnilyTW09} give an exact linear programming formulation for the multi-item case in which holding costs are linear but not identical; however, there is a priority on items so that the holding costs $h_{1t} \geq h_{2t} \geq \cdots$ for all time periods $t$.  Additionally, they assume that it is possible to place multiple orders at time $t$, each of capacity $U_t$ and cost $f_t$.  Carnes and Shmoys \cite{CarnesS15} give a primal-dual 2-approximation algorithm for the multi-item case in which holding costs are identical and linear.  Even, Levi, Rawitz, Schieber, Shahar, and Sviridenko \cite{EvenLRSSS08} give another dynamic program for the multi-item case with special assumptions on the holding costs.  Levi, Lodi, and Sviridenko \cite{LeviLS08} give an LP-rounding 2-approximation algorithm for the multi-item case in which the capacities are uniform and holding costs are increasing.  Li \cite{Li17} gives an LP-rounding 10-approximation algorithm for the multi-item case with the same holding costs when capacities are non-uniform.

In this paper, we will consider another special case of transportation cost for the capacitated facility location problem, and show that it subsumes the multiple-item lot-sizing problem with linear and identical holding costs.  In particular, we consider the case of transportation costs that obey the {\em Monge property}: assuming that $F=\{1,\ldots,m\}$, $D = \{1,\ldots,n\}$, the transportation costs $c_{ij}  \in \Re^{\geq 0} \cup \{\infty\}$ obey the Monge property if
\begin{equation} \label{eq:monge}
c_{hj} + c_{ik} \leq c_{hk} + c_{ij} \qquad \mbox{ for } 1 \leq h < i \leq m, 1 \leq j < k \leq n,
\end{equation} where $a + \infty = \infty$ for any $a$.  Hoffman \cite{Hoffman63} identified the Monge property of transportation costs as allowing for simple greedy algorithms to find optimal solutions in many cases.  Burkard, Klinz, and Rudolph \cite{BurkardKR96} give a nice survey of the Monge property and its various applications. 

Crucially, for our purposes, if transportation costs obey the Monge property, then the following {\em transportation problem} can be solved via a greedy algorithm: given supplies $s_i \geq 0$ for $i = 1,\ldots,m$ and demands $d_j \geq 0$ for $j = 1,\ldots,n$ such that $\sum_{i=1}^m s_i = \sum_{j=1}^n d_j$, we wish to find a transshipment $x_{ij} \geq 0$ for $i = 1, \ldots, m$ and $j = 1,\ldots,n$ that minimizes $\sum_{i=1}^m\sum_{j=1}^n c_{ij} x_{ij}$ such that $\sum_{i=1}^m x_{ij} = d_j$ for each $j=1,\ldots,n$ and $\sum_{j=1}^n x_{ij} = s_i$ for each $i=1,\ldots,m$.  In particular, Algorithm \ref{alg:monge} finds an optimal solution to the problem when the transportation costs have the Monge property; the algorithm visits the clients and facilities in increasing order, and greedily sends as much demand as possible from the current client to the current facility.

\begin{algorithm}[t]
$i \gets 1$\;
$j \gets 1$\;
$x \gets 0$\;
\While{$i \leq m$ and $j \leq n$}{
	\eIf{$s_i \geq d_j$}{
		$x_{ij} \gets d_j$\;
		$j \gets j + 1$\;
		$s_i \gets s_i - d_j$\;		
	}{
		$x_{ij} \gets s_i$\;
		$i \gets i + 1$\;
		$d_j \gets d_j - s_i$\;
	}
}
\Return $x$\;
	\caption{An algorithm for finding a solution to a transportation problem with Monge costs.}
	\label{alg:monge}
\end{algorithm}

In the case of the capacitated single-item lot-sizing problem with linear holding costs, it is easy to show a reduction to Monge costs.  Our ordering on orders (facilities) is in order of time, as is our ordering on demands (clients).  For simplicity, we reindex the orders and demands to have an index corresponding to their time period, so that the cost of an order at time $t$ is $f_t$, the demand at time $t$ is $d_t$, the capacity of an order that can be placed at time $t$ is $U_t$, the holding cost per unit is $h_t$ for the time period $[t,t+1)$, and the cost of holding a unit of an item from time $t$ to time $t'$ is denoted $c_{t,t'} = \sum_{k=t}^{t'-1} h_k$.  Note that the cost for meeting a demand $d_t$ from an order at time $t' > t$ is infinite, so that $c_{t',t} = \infty$.  To verify the Monge inequality, we need to pick four time periods, call them $p < q$ (for the orders) and $s < t$ (for the demands) and verify that
$$c_{ps} + c_{qt} \leq c_{pt} + c_{qs}.$$  If $s<p$ then $c_{ps} = \infty$, but then $s < p < q$ and $c_{qs} = \infty$.  If $t < q$ then $c_{qt} = \infty$, but then $s < t < q$ so that $c_{qs} = \infty$.  So we assume that $p \leq s < t$ and $q \leq t$, so that 
$$c_{ps} + c_{qt} = \sum_{k=p}^{s-1} h_k + \sum_{k=q}^{t-1} h_k \leq c_{pt} + c_{qs} =  \left\{
\begin{array}{ll}
\infty & \mbox{if } s < q \\
\sum_{k=p}^{t-1} h_k + \sum_{k=q}^{s-1} h_k & \mbox{if } q \leq s.
\end{array}
\right.$$ It is easy to give a similar reduction for capacitated multi-item lot-sizing with linear and identical holding costs.

We also observe that the well-studied single-demand capacitated facility location problem \cite{CarrFLP00,CarnesS15,CheungW17}, in which $|D|=1$ is also a capacitated facility location problem whose transportation costs obey the Monge property since (\ref{eq:monge}) holds trivially.

Our results are as follows.  We first show in Section \ref{sec:pb} that a straightforward dynamic program gives a polynomial-time algorithm for the capacitated facility problem with costs obeying the Monge property in which demands are polynomially bounded.  The main issue to overcome relative to the easy algorithm above for the underlying transportation problem is that we do not know which facilities are opened, and it may be the case that the total open capacity exceeds the total demand.  However, by using dynamic programming we can guess the open facilities, and we can guess how much demand is served by each open facility, and this guessing can be done in polynomial time as long as the total demand is polynomially bounded.   In the case that demands are not polynomially bounded, we extend the algorithm of Van Hoesel and Wagelmans \cite{VanHoeselW01} to give an FPTAS for the capacitated facility location problem with costs obeying the Monge property in Section \ref{sec:npb}.   We further discuss some extensions of our dynamic program to cases in which $c_{ij}$ is finite only for a fixed range of client $j$.

By the reduction above, we immediately get a polynomial-time algorithm for capacitated lot-sizing problems with linear holding costs and polynomially bounded demands, and the FPTAS of Van Hoesel and Wagelmans in the case that demands are not polynomially bounded.  We obtain the same for the single-demand capacitiated facility location problem.    Given that it is NP-hard to approximate even the uncapacitated facility location in the metric case to within a factor less than 1.463, we see that costs having the Monge property give rise to a particularly easy case of the capacitated facility location problem.

\section{The Dynamic Program for Polynomially-Bounded Demands}
\label{sec:pb}

In this section, we show that a simple dynamic program gives the optimal solution to the capacitated facility location problem when the transportation costs obey the Monge property.  The basic idea is quite simple: for each facility $i \in F$, we decide whether to open facility $i$, and, if so, how much demand to serve from it, up to its capacity $U_i$ or the total demand remaining.  Because of the Monge property, we know that we can assign the demand from clients to the facility in a greedy manner, and it is easy to state what the remaining unassigned demand will be.  We assume that clients and facilities are indexed by $1,\ldots,n$ and $1,\ldots,m $ (respectively) so that the Monge equation (\ref{eq:monge}) holds.

In particular, we will use the function $C(i,j,d)$ for the dynamic program, in which $d \leq d_j$.  The function $C(i,j,d)$ is the cost of the optimal solution to the capacitated facility location problem having only facilities $i$ and higher, only clients $j$ and higher, in which there are only $d$ units of demand to be served at client $j$. If we reach the end of the list of facilities, and have not served all the demand, then the cost is infinite, so that $C(m+1,j,d) = \infty$ if $d + \sum_{k =j+1}^n d_k > 0$.  If we reach the end of the clients, and there is no more demand to be served, then the cost is zero, so that $C(i,n+1,d) = 0$, and we assume $d = d_{n+1} = 0$.  To compute the optimal solution, we wish to find the cost of $C(1,1,d_1)$.

\newcommand{\NC}{\mbox{\sc NC}}
\newcommand{\DR}{\mbox{\sc DR}}
\newcommand{\TR}{\mbox{\sc TC}}

In order to state the recurrence relation, it will be useful to have some helper functions.  Given that we can assume that demand will be served greedily, if facility $i$ will serve $u$ units of demand, and we need to serve $d$ units from client $j$, then these $u$ units will serve all clients from $j$ to client $\ell$, where $\ell$ is the first index such that $d + \sum_{k = j+1}^\ell d_k > u$ (note that if $d > u$, then $\ell = j$). Let $\NC(i,u,j,d)$ return the index $\ell$.  Additionally, we need to know how many units of demand of client $\ell$ remain to be served.  For example, if $d > u$ and $\ell = j$, then $d-u$ units of demand remain to be served at client $\ell = j$.  More generally, if $\ell$ is the smallest index such that $d + \sum_{k =j+1}^\ell d_k > u$, then $d_\ell - (u - (d + \sum_{k =j+1}^{\ell - 1}d_k))$ remain to be served at client $\ell$.  Let $\DR(i,u,j,d)$ denote this amount of demand that remains to be served at client $\ell$ after the $u$ units of demand have been served from facility $i$.  Finally, we let $\TR(i,u,j,d)$ be the transportation cost of serving $u$ units of demand from facility $i$ by greedily assigning demand to $i$, so that if $\ell = \NC(i,u,j,d)$ and $d' = \DR(i,u,j,d)$
$$\TR(i,u,j,d) = \left\{ \begin{array}{ll}
c_{ij} \cdot (d - d') & \mbox{if } \ell = j, \\
 c_{ij} \cdot d + \sum_{k = j+1}^{\ell - 1} c_{ik} \cdot d_k + c_{i\ell} \cdot (d_\ell - d') & \mbox{otherwise.}
 \end{array}\right.$$

We can now write the recurrence relation.  In it, for facility $i$, we optimize over whether $i$ remains unopened (in which case remaining demand gets served from facilities $i+1$ and higher), or, if $i$ is opened, how many units of demand are served from $i$.  The relation is below:
\begin{align}
C(i,j,d) & =  \min\left(C(i+1,j,d), \right.  \label{dp-pb} \\
	&   \left. f_i + \min_{u: 1 \leq u \leq \min(U_i, d + \sum_{k =j+1}^n d_j)}\left(\TR(i,u,j,d) + C(i+1,\NC(i,u,j,d), \DR(i,u,j,d)\right)\right) \nonumber 
\end{align}
We observe that as long as $u$ is always bounded by a polynomial in the input size, then the dynamic program itself can be executed in polynomial time.  Thus if the demands are polynomially bounded, $u$ is polynomial in the input size, and the algorithm runs in polynomial time.

It is now straightforward to prove that the dynamic program finds the optimal solution to the problem.

\begin{theorem}
The dynamic program given in (\ref{dp-pb}) computes the optimal cost of $C(i,j,d)$.
\end{theorem}

\begin{proof}
We show that one of the options considered in the dynamic program corresponds to an optimal solution for the problem instance for $C(i,j,d)$, and this is sufficient.  We assume by induction that we have the optimal cost for $C(i',j',d')$ for $i' > i$ or $j' > j$ or $d' < d$.  In a given optimal solution to the instance having only facilities $i$ and higher, only clients $j$ and higher, and only demand $d$ at client $j$, either facility $i$ is opened, or it is not.  If it is not opened, then the demand from clients $j$ and higher must be met by facilities $i+1$ and higher, and $C(i+1,j,d)$ gives the optimal cost to this problem. 

If $i$ is opened in the given optimal solution, then it serves some amount of demand $u$, where $u \geq 1$, $u \leq U_i$, and $u \leq d + \sum_{k =j+1}^n d_k$.  We now claim that the $u$ units of demand can be assigned greedily to the facility $i$.  To prove this claim, we fix an optimal solution that assigns as much of the $u$ units of demand to facility $i$ from clients of index as small as possible. Suppose that we run the greedy algorithm to assign $u$ units of demand to $i$, and $j' \geq j$ is the lowest indexed client such that the fixed optimal solution assigns fewer units of demand from $j'$ to $i$ than the greedy algorithm; instead the optimal solution assigns at least one unit of demand from $j'$ to a facility $h$ with $h > i$; clearly this facility must be open in the optimal solution.  Since we know that both the greedy algorithm and the optimal solution assign $u$ units of demand to facility $i$, if the optimal solution assigns fewer units of the demand of client $j'$ than the greedy algorithm to $i$, then there must be some other client $k > j'$ such that the optimal solution assigns at least a unit of demand from $k$ to $i$.  However, because we know that the Monge inequality (\ref{eq:monge}) holds (since $j' < k$ and $i < h$), we can reassign the optimal solution assignment of demand from $j'$ to $h$ and $k$ to $i$, to instead be from $j'$ to $i$ and $k$ to $h$ without increasing the cost.  Thus we have another optimal solution in which demand is assigned to $i$ from a lower indexed client than before, which contradicts our assumption about the optimal solution.  Thus we can assume that the $u$ units of demand is assigned greedily to facility $i$, and so the cost of assigning the $u$ units of demand, and the remaining problem is as given in the recurrence relation.
\end{proof}

\section{The Dynamic Program for Non-Polynomially Bounded Demands}
\label{sec:npb}

\subsection{The Dynamic Program}

In this section we give a dynamic program for the case of non-polynomial demands.  Our dynamic program follows ideas of Van Hoesel and Wagelmans \cite{VanHoeselW01} for the capacitated single-item lot-sizing problem.

Let $z^*$ be the cost of an optimal solution to the capacitated facility location problem with Monge costs. We first assume that we know a value $B \in \mathbb{Z}$ such that $z^* \leq B \leq m \cdot z^*$; we will show how to find such a $B$ in polynomial-time in Section \ref{sec:findB}. Given $B$, we present a non-polynomial-time dynamic program that finds the optimal solution to the facility location problem. We then show how to approximate this dynamic program to obtain an FPTAS. 

For $i \in F$ and $b \in \{1,2, \ldots, B\}$, our value function $V_{i}(b)$ will be the maximum demand we can meet by only opening facilities $i$ and higher, spending at most budget $b$, and fulfilling demand such that if we satisfy any demand for client $j$ then we satisfy all demand for clients $j' > j$. If we reach the end of our facilities, then we cannot fill any more demand. Thus $V_{m+1}(b) = 0$. To find the optimal solution cost, we simply want to find the smallest value $b$ such that $V_{1}(b) \geq \sum_{j} d_j$.  

Before we state the recurrence relation, we first define a helper function. As observed in the previous section, the Monge property allows us to assume that demand is filled from the last client forward. Thus if facilities of index larger than $i$ fulfill $d$ units of demand using budget $b' \leq b$, then we can find the remaining demand for all clients. Let $\bar{d}_j(i,d)$ denote this remaining demand. Further, this implies that $i$ spends the remaining budget $b-b'$ serving its demand. In particular, we can first subtract $f_i$ for opening $i$ and then assign $i$ to fulfill the remaining demand $\bar{d}_j(i,d)$ starting from $j = n$ until we run out of budget or capacity. Let $DM(i, d, b-b')$ denote the amount of demand met. 

We can now write the recurrence relation. 
\begin{equation}
\label{eq:dp}
V_{i}(b)  = \max_{0 \leq b' \leq b} \left(V_{i+1}(b') + DM(i, V_{i+1}(b'), b-b') \right).
\end{equation}

\begin{theorem}
	\label{thm:dp_opt}
	The dynamic program given in Equation~(\ref{eq:dp}) correctly computes $V_{i}(b)$. 
\end{theorem}
\begin{proof}
	Recall that $V_{i}(b)$ is the maximum number of units of demand met when only facilities $i$ and higher are opened, all demand is satisfied right to left (from $j=n$ backward), and we spend at most budget $b$. The base case holds by construction. We assume by induction that the recurrence relationship correctly computes $V_{i'}(b')$ for all $i' > i$ and $b' \leq b$.  In a given optimal feasible solution that only opens facilities $i$ and higher, and spends budget at most $b$, the amount of budget used by facilities $i' \geq i$ is $b'$ for some $0 \leq b' \leq b$. 
	
	Suppose that the optimal solution assigns $i$ to serve a unit of demand from client $k$ and assigns $i' >i$ to serve client $k' < k$. Then, by the Monge property, swapping this assignment so that $i$ serves $k'$ and $i'$ serves $k$ cannot increase the total cost of the solution. Therefore, the solution remains feasible, and we can assume that if a facility $i' \geq i$ serves client $j$ then all clients of index larger than $j$ are served fully by facilities $i'$ and higher. Thus, by induction, facilities $i'$ for $i' \geq i$ serve $V_{i+1}(b')$ units of demand from right to left. Further, facility $i$ can then use the remaining budget $b'-b$ to serve as much demand as possible continuing from right to left. The amount served by $i$ is given by $DM(i,V_{i+1}(b'),b'-b)$. Overall, this shows that the maximum amount of demand satisfied is 
	\[  V_{i+1}(b') + DM(i, V_{i+1}(b'), b-b') . \]
	
	Taking the maximum over all feasible $b'$ yields the recurrence relationship.
\end{proof}

The problem with the value functions $V_{i}(b)$ is that there are a non-polynomial number of possible values for $b$. To get around this, we will fix an integer $K  = \max(1, \lceil \frac{\epsilon B}{m(m+1)} \rceil )$ and restrict the budgets to be multiples of $K$ in the set $\mathcal{B}_K =  \{0,K, 2K, \ldots, (\lceil B/K \rceil + m) K \}$. Let $\lceil x \rceil _K$ ($\lfloor x \rfloor_K$) return $x$ rounded up (down) to the nearest multiple of $K$. Then we modify the DP recurrence relationship using these restricted budgets. In particular, for $i \in F$, $j \in D$, and $b \in \mathcal{B}_K$, we let
\begin{equation}
\label{eq:dp_rounded}
\bar{V}_{i}(b)  = \max_{b' \in \mathcal{B}_K: 0 \leq b' \leq b} \left( \bar{V}_{i+1}(b') + DM(i, \bar{V}_{i+1}(b'), b-b') \right). 
\end{equation}
The base cases remain the same: $\bar{V}_{m+1}(b) = 0$. 

\begin{lemma}
	\label{lem:dp_grid}
	If $V_{i}(b)= d$, then $\bar{V}_{i}(\lceil b+(m-i+1) K \rceil _K )  \geq d$. 
\end{lemma}
\begin{proof}
	The base cases hold by construction. We assume by induction that the lemma holds for all $V_{i'}(b')$ for all $i' > i$  and $b'\leq b$. Let $b^* = \argmax_{0 \leq b' \leq b} (V_{i+1}(b') + DM(i, V_{i+1}(b'), b-b'))$. Then our $K$-spaced grid implies that 
	\begin{equation}
	\label{eq:budget}
	\lceil b+(m-i+1) K \rceil _K - \lceil b^* + (m-i) K \rceil_K \geq b-b^*. 
	\end{equation}
	Thus 
	\begin{align*}
	\bar{V}_{i+1}(\lceil b+(m-i+1) K \rceil _K) &  \geq \bar{V}_{i+1}( \lceil b^* +(m-i) K \rceil_K) + DM(i, \bar{V}_{i+1}(\lceil b^* + (m-i) K \rceil_K), b-b^*) \\
	& \geq V_{i+1}(b^*) + DM(i, V_{i+1}(b^*), b-b^*) \\
	& = V_{i}(b).
	\end{align*}
	The first inequality follows by Equations~(\ref{eq:dp_rounded}) and (\ref{eq:budget}). The second inequality follows by the inductive hypothesis, and the final equality follows from our definition of $b^*$.  Thus, overall, 
	\[ \bar{V}_{i}(\lceil b+(m-i+1) K \rceil _K )  \geq  V_{i}(b) = d. \]
\end{proof}

\begin{theorem}  Let $z^*$ be the cost of an optimal solution to the instance of the capacitated facility location problem.  Then 
	in polynomial time we can find a feasible solution with cost 
	\[ z_A \leq (1+\epsilon) z^* . \] 
\end{theorem}
\begin{proof}
	Recall that $B$ is an integer such that $z^* \leq B \leq m\cdot z^*$. By Theorem~\ref{thm:dp_opt} and Lemma~\ref{lem:dp_grid}, $\bar{V}_{1}(\lceil z^* +mK\rceil_K)  \geq V_1(z^*) \geq \sum_{j=1}^n d_j$.    Thus through the modified DP we can find a feasible solution such that the cost is at most 
	\[ \lceil z^* + mK \rceil_K \leq z^* + (m+1) \max \left (1, \Big \lceil \frac{\epsilon B}{m(m+1)} \Big \rceil \right ) \leq z^* + \epsilon z^*. \]
	Further, finding this solution requires computing $m \cdot ( \lceil B/K\rceil + m+1) = O( m^3 \epsilon^{-1})$ values for the dynamic program. To compute each $\bar{V}_i(b)$, we need to maximize over $b'$ where computing each $DM(i, \bar{V}_{i+1}(b'), b-b')$ takes at most $O(n)$ time. Thus finding the maximum $b'$ takes $O(n \cdot ( \lceil B/K\rceil + m+1)) = O(n m^2)$ time. This yields an overall runtime of $O(n m^5 \epsilon^{-1})$. 
\end{proof}

\subsection{Upper Bound on an Optimal Solution}
\label{sec:findB}

We now show how we can find $B \in \mathbb{Z}$ such that $z^* \leq B \leq mz^*$. To do so, we will present a simpler dynamic program that for fixed $\ell$ finds whether there exists a feasible solution in which each facility contributes at most $\ell$ to the total cost. In other words, for all $i \in F'$, 
\[ f_i + \sum_{j \in D} c_{i,j} d_j x_{i,j} \leq \ell. \]
If we can find the minimum $\ell^*$ such that a feasible solution exists, then we know that $z^* \leq m\ell^* \leq mz^*$ and we can set $B = m\ell^*$.

For a fixed value $0 \leq \ell \leq \max_{i} (f_i + \sum_{j \in D: c_{i,j} < \infty} c_{i,j} d_j)$, we will use a dynamic program to determine whether there exists a feasible solution in which each facility contributes at most $\ell$ to the total cost. Similar to our previous DP, for all $i \in F$ our value function $\hat{V}^\ell_i$ will be the maximum demand that can be met right to left by only opening facilities $i$ and larger where each facility contributes at most cost $\ell$. If we reach the end of our facilities, then we cannot meet any more demand $\hat{V}^\ell_{m+1} = 0$. 

In general, we use the following recurrence relationship
\begin{equation}
\label{eq:dp_bound}
\hat{V}^\ell_{i}  = \hat{V}^\ell_{i+1} + DM(i, \hat{V}^\ell_{i+1}, \ell)
\end{equation}
To find whether a feasible solution exists, we then check whether or not $\hat{V}^\ell_{1} \geq \sum_{j=1}^n d_j$. 

\begin{lemma}
	The dynamic program in Equation~(\ref{eq:dp_bound}) correctly computes $\hat{V}^\ell_{i}$.
\end{lemma}
\begin{proof}
	The proof of correctness follows that in Theorem~\ref{thm:dp_opt}. 
\end{proof}

\begin{corollary}
	In polynomial time we can find an integer $B$ such that $z^* \leq B \leq mz^*$.
\end{corollary}
\begin{proof}
	Since $\ell^*$ is integer-valued, we can find this value through binary search on the range $[1, \max_{i}( f_i + \sum_{j \in D: c_{i,j} < \infty} c_{i,j}d_j)]$. We then set $B= m\ell^*$. 
\end{proof}

\subsection{Extension of the Dynamic Program}

In this section, we extend the DP to incorporate release dates or limits on the transportation costs. In particular, we assume that for each client $j$ there is a release date $r_j$ and deadline $t_j$ such that $c_{i,j} = \infty$ if $i < r_j$ or $i>t_j$. We first consider the case that $r_1 \leq r_2 \leq ...\leq r_n$ and $t_1 \leq t_2 \leq ... \leq t_n$. Suppose that for all facilities $i<i'$ and clients $j<j'$ such that $i$ and $i'$ can both serve $j$ and $j'$ with non-infinite cost 
\[c_{i,j}+c_{i',j'} \leq c_{i,j'}+c_{i',j}. \]
In other words, the costs are Monge when release dates and deadlines are met (ex. linear holding costs within release date and deadline). Then for all facilities $i<i'$ and $j<j'$ 
\[c_{i,j}+c_{i',j'} \leq c_{i,j'}+c_{i',j} \]
since the right-hand side of the equation is non-infinite if and only $i \geq r_{j'} \geq r_{j}$ and $i'\leq t_{j}\leq t_{j'}$. Thus the transportation costs are Monge and we can use the DP as above.

We can also slightly relax this assumption about the release dates by modifying the DP. In particular, we assume that we can separate the clients into two disjoint subsets $S_1$ and $S_2$ such that $r_j = 1$ for all $j \in S_1$ and $r_j \leq r_{j'}$ for all $j, j' \in S_2$ such that $j \leq j'$. In other words, $S_2$ contains the clients with time-sensitive demands. Then it is clear that holding costs are Monge within both subsets. That is, for $k=1,2$, for all facilities $i<i'$ and clients $j<j'$ such that $j,j' \in S_k$
\[c_{i,j}+c_{i',j'} \leq c_{i,j'}+c_{i',j} .\]
This idea allows us to extend the DP ideas from above. As before, let $B$ be an upper bound on the cost of an optimal solution, and let $\mathcal{B}_K =  \{0,K, 2K, \ldots, (\lceil B/K \rceil + m) K \}$ where $K  = \max(1, \lceil \frac{\epsilon B}{m(m+1)} \rceil )$.

For $i \in F$ and $b_0, b_1,b_2 \in \mathcal{B}_K$, our value function $V_{i}(b=[b_0,b_1,b_2])$ will be the maximum demand $d =(d_1,d_2)$ from $S_1$ and $S_2$ respectively we can meet by only opening facilities $i$ and higher, spending at most budget $b_0$ on opening costs and at most $b_k$ on serving subset $S_k$ for $k =1,2$, and fulfilling demand such that if we satisfy any demand for client $j \in S_k$ for $k \in \{1,2\}$ then we satisfy all demand for clients $j' \in S_k$ such that  $j'> j$. If we reach the end of our facilities, then we cannot fill any more demand. Thus $V_{m+1}(b_0,b_1,b_2) = 0$. To find the optimal solution cost, we simply want to find the smallest value $b = b_0+b_1+b_2$ such that $V_{1}(b_0,b_1,b_2) \geq \sum_{j} d_j$.  

Before we state the recurrence relation, we first define a helper function. For any facility $i$ and budget $b = [b_0,b_1,b_2]$, let budget $b' =[b'_0, b'_1,b'_2] \in \mathcal{B}_K^3$ such that $b' \leq b$. 
Suppose that facilities $>i$ satisfy demand $d=(d_1,d_2)$ from $S_1$ and $S_2$ respectively given budget $b'$. 
Then this leaves budget $b-b'$ for facility $i$. If $f_i \leq b_0-b_0'$, then we can open facility $i$. 
Next, facility $i$ has budget $b_k-b'_k$ to serve each $S_k$. Given that within each $S_k$ we can assume that demand is filled from the last client forward, then we serve the remaining demand right to left using these budgets until either we use up the budget or we reach the maximum capacity. Without loss of generality, we may assume that we serve clients in $S_1$ first using budget $b_1-b_1'$ before serving clients in $S_2$, if capacity remains. Let $DM(i, d, b-b')$ denote the total amount of demand met represented as a vector $(d_1',d_2')$ such that $d'_k$ is the amount of demand met from $S_k$ by facility $i$. 

We can now write the recurrence relation. 
\begin{equation}
\label{eq:dp2}
V_{i}(b)  = \max_{b' \in \mathcal{B}_K^3: b' \leq b} \left(V_{i+1}(b') + DM(i, V_{i+1}(b'), b-b') \right).
\end{equation}
The correctness of this DP follows from the same arguments as previously.

\subsection*{Acknowledgments}

Part of this work was carried out at the Simons Institute for the Theory of Computing.

\bibliographystyle{abbrv}
\bibliography{monge}

\end{document}